\documentclass[submission,copyright,creativecommons]{eptcs}
\providecommand{\event}{GANDALF 2019}
\usepackage{enumitem}
\usepackage{stmaryrd,amsfonts,amsmath,amssymb,amsthm}
\usepackage{etoolbox}
\usepackage{graphicx}
\usepackage{float}
\usepackage{multicol}
\usepackage[curve,cmtip]{xypic}
\usepackage{titlesec}
\titlespacing*{\section}{0pt}{0.5\baselineskip}{0.25\baselineskip}
\titlespacing*{\subsection}{0pt}{0.4\baselineskip}{0.2\baselineskip}

\theoremstyle{plain}
\newtheorem{theorem}{Theorem}[section]
\newtheorem{proposition}[theorem]{Proposition}
\newtheorem{definition}[theorem]{Definition}
\AfterEndEnvironment{theorem}{\noindent\ignorespaces}
\AfterEndEnvironment{proposition}{\noindent\ignorespaces}
\AfterEndEnvironment{definition}{\noindent\ignorespaces}

\theoremstyle{remark}
\newtheorem{remark}[theorem]{Remark}
\newtheorem{example}[theorem]{Example}
\newtheorem{ass}[theorem]{Assumption}

\makeatletter
\renewcommand{\paragraph}{%
  \@startsection{paragraph}{4}%
  {\z@}{0.75ex \@plus 1ex \@minus .2ex}{-1em}%
  {\normalfont\normalsize\bfseries}%
}
\makeatother

\renewcommand{\coprod}{\mathop{\text{\fakecoprod}}}
\newcommand{\fakecoprod}{%
  \sbox0{$\prod$}%
  \smash{\raisebox{\dimexpr.9625\depth-\dp0}{\scalebox{1}[-1]{$\prod$}}}%
  \vphantom{$\prod$}%
}

\newcommand{\pto}{\to\hspace{-7pt}\shortmid \hspace{7pt}}
\newcommand{\dom}{{\mathsf{dom}}}

\newcommand{\FPow}{{\mathcal{P}_\omega}}
\newcommand{\A}{{\mathcal{A}}}

\newcommand{\R}{{\mathcal{R}}}
\newcommand{\G}{{\mathcal{G}}}
\newcommand{\Runs}{{\mathsf{Runs}}}
\newcommand{\val}{{\mathsf{val}}}

\newcommand{\arity}{{\mathsf{ar}}}
\newcommand{\ext}{{\mathsf{ext}}}
\newcommand{\ran}{{\mathsf{ran}}}
\newcommand{\supp}{{\mathsf{supp}}}

\newcommand{\Id}{{\mathsf{Id}}}
\renewcommand{\phi}{\varphi}

\newcommand{\T}{{\mathsf T}}
\newcommand{\Set}{{\mathsf{Set}}}

\title{Resource-Aware Automata and Games for Optimal Synthesis}
\author{Corina C\^{\i}rstea
\institute{University of Southampton, UK}
\email{cc2@ecs.soton.ac.uk}}
\def\titlerunning{Resource-Aware Automata and Games for Optimal Synthesis}
\def\authorrunning{C. C\^{\i}rstea}

\begin{document}
\maketitle
\abovedisplayskip=5pt
\belowdisplayskip=5pt

\begin{abstract}
We consider quantitative notions of parity automaton and parity game aimed at modelling resource-aware behaviour, and study (memory-full) strategies for exhibiting accepting runs that require a minimum amount of initial resources, respectively for winning a game with minimum initial resources. We also show how such strategies can be simplified to consist of only two types of moves: the former aimed at increasing resources, the latter aimed at satisfying the acceptance condition.
\end{abstract}

\section{Introduction}

This paper studies quantitative versions of B\"uchi/parity automata on infinite words/trees, as well as of B\"uchi/parity games \cite{GraedelTW}, aimed at the modelling, verification and synthesis of \emph{resource-aware} systems. In the case of automata, our results concern synthesising strategies for exhibiting \emph{accepting runs} that are (i) \emph{resource-aware}, in the sense that at any point during the run, sufficient resources are available to take the next transition, and (ii) \emph{optimal} in their requirement on initial resources. In the case of games, we are interested in winning strategies that require the smallest amount of initial resources.

Our approach builds on earlier work on a general approach to quantitative verification \cite{Cirstea17a,Cirstea14,CirsteaSH17,CirsteaLMCS}. This work uses  \emph{coalgebras} \cite{JacobsBook} to model systems with quantitative features and quantitative fixpoint logics to specify correctness properties, and extends the standard automata-based approach to verification to this setting. This general approach is parameterised by (i) a quantitative branching type (modelled via a partial semiring of quantities, e.g.~probabilities or resources) and (ii) a qualitative behaviour type (e.g.~infinite words or trees): each choice for these two ingredients gives rise to a specific model type together with an associated quantitative logic, for which the model checking problem reduces to computing the \emph{extent} (generalising non-emptiness) of a certain quantitative parity automaton \cite{CirsteaSH17}. At the heart of our approach are nested fixpoints, akin to those used in the semantics of the modal $\mu$-calculus.

Here we focus on a concrete instantiation for the universe of quantities, with the resulting quantitative automata being used to model resource-aware behaviour. The qualitative behaviour type remains a parameter, thus making our results applicable to automata over both \emph{words} and \emph{trees}. To maintain this level of generality, and at the same time be succinct in our presentation, we continue to rely on a coalgebraic formulation of key concepts; however, the use of coalgebraic machinery is minimised, to make our results accessible also to readers not familiar with coalgebra.

The notion of resource-aware automaton we consider is similar to that of \cite{CirsteaLMCS}, and more general than the one in \cite{CirsteaSH17}, in that not only resource \emph{use} but also resource \emph{gain} is modelled. Unlike non-deterministic parity automata/games where strategies for exhibiting accepting runs\,/\,winning a game are memoryless \cite{Jurdzinski00}, we show that, in the resource-aware case, \emph{memory-full strategies} are needed even for automata. This is because the synthesised run must not only satisfy the parity condition, but also respect available resources; as a result, the transition to take in a state depends on the available resources, being aimed either at accumulating additional resources or at advancing towards an accepting state.

To study such memory-full strategies, we exploit a connection to standard parity games. Specifically, we associate to each resource-aware parity automaton/game a standard parity game, and show how memory-full strategies for exhibiting optimal accepting runs in the original automaton can be derived from memoryless winning strategies in the associated parity game. Further, we show how to reduce these strategies to only \emph{two} moves per state while maintaining their winning property. Finally, we show how the algorithm for computing the extent of a quantitative automaton \cite{CirsteaSH17} can be enhanced to compute such strategies.

While the majority of the paper concerns automata, we also consider a notion of quantitative parity game, parameterised by a semiring of quantities, and show how the notion of extent generalises from automata to games. Moreover, in the resource-aware case, we show that the construction of strategies for exhibiting optimal accepting runs also generalises -- where this time, the objective is to synthesise optimal winning strategies for player $\exists$.

Our algorithms for computing optimal strategies apply to resource-aware \emph{B\"uchi} automata/games on words/trees; a generalisation to \emph{parity} automata is discussed briefly, with a full treatment  deferred to future work. In the B\"uchi case, the computation of optimal strategies has complexity quadratic in the size of the automaton/game, \emph{assuming extents have already been computed}. (The computation of extents has complexity $O(n^{k+1})$ with $n$ the automaton size and $k$ the number of parities.)

Our results yield techniques for the synthesis of behaviours within specified design spaces (automaton case), or the quantitative synthesis of strategies in non-cooperative settings (game case), in scenarios where both correctness and optimality are important. Similar combinations of qualitative and quantitative objectives have been studied previously. \cite{CHJ05} considered games that combine parity objectives with mean-payoff ones, and studied algorithms for computing the values (for each of the two players) associated to such games along with strategies to achieve them. (Unlike parity or mean-payoff games, such strategies are not anymore memoryless; in fact they may require \emph{infinite} memory.) \cite{Chatterjee12} considered weighted parity games, with weights thought of as energy values and with objectives combining a parity condition with the requirement to maintain a positive energy level. An algorithm for solving such games with complexity exponential in the size of the game and linear in the largest weight was described, and it was shown that \emph{finite} (exponential) memory suffices to represent optimal strategies. \cite{FZ14d} combined parity/Streett conditions with quantitative constraints bounding the costs between requests (odd parity states) and responses (higher, even parity states), and showed that solving such games can be reduced to solving several instances of their purely qualitative variants. A generalisation from (positive) costs to weights that can be either positive or negative was described in \cite{SWZ18}, and shown to be polynomial-time equivalent to the energy parity games of \cite{Chatterjee12}. Finally, \cite{Bouyer08} considered weighted \emph{timed} games with energy constraints, where positive/negative weights can be associated to both locations and transitions, and showed that the existence of strategies ensuring that the accumulated weight in any finite prefix of the resulting play stays within a given interval is undecidable.

While our work addresses a similar problem, a key advantage compared to existing work is the applicability of our results to a \emph{wider class of qualitative behaviours}: not just infinite words, but also infinite trees or a combination of word- and tree-like behaviour, including terminating behaviour. Such an extension allows modelling systems whose structure varies over time (e.g.~new processes being created and other processes terminating while the system executes), where now the goal is to synthesise an optimal \emph{system} w.r.t.~the \emph{overall} amount of resources required. Our use of coalgebras as models allows for an increased flexibility in the type of qualitative behaviour such systems can exhibit; whereas our use of coinduction on a technical level supports dealing with such general models of behaviour (where, for example, the standard notion of a cycle, as used e.g.~in \cite{Chatterjee12}, becomes more complex).

Our setting is closest in spirit to the work on \emph{energy games} \cite{Chatterjee12}, and our results on reduced-memory strategies (Section~\ref{sec:reduced-mem-strat}) are inspired by similar results in loc.\,cit.~(although our proofs are different from those in \cite{Chatterjee12}); we expand on the connection to \cite{Chatterjee12} later in the paper (see Remark~\ref{rem:chat}). The main \emph{technical} difference between our approach and existing work is that we exploit a quantitative generalisation of the nested fixpoints used to decide non-emptiness of parity automata\,/\,the winner of a parity game. This requires monotonicity/continuity assumptions, to allow a generalisation of the fixpoint approach to semiring-valued maps. Our choice to associate resource gains to automaton states and resource usage to automaton transitions (see Definition~\ref{def:wpao}) ensures that the required assumptions are satisfied. Indeed, this particular choice of model structure is driven by the need to satisfy these assumptions.

The paper is structured as follows. After some basic preliminaries (Section~\ref{sec:prel}), we recall the notion of \emph{weighted parity automaton with offsetting} and the associated notion of \emph{extent} \cite{CirsteaLMCS}, and show how one can associate a \emph{value} to each run of such an automaton, which in the resource-aware instantiation measures the minimal initial resources required to to execute the run (Section~\ref{sec:wpao}). The main contributions are:
\begin{enumerate}[noitemsep]
\item a study of memory-full strategies for exhibiting optimal accepting runs in resource-aware \emph{B\"uchi} automata over words/trees, including a simplified version of such strategies that uses only two moves for each state, along with an algorithm for computing such strategies (Section~\ref{sec:strat-aut}),
\item a notion of weighted parity game with offsetting, and a generalisation of the notion of extent to such games (Section~\ref{sec:games}),
\item a brief study of winning strategies for the resource-aware instance of weighted parity games with offsetting, outlining how the results of Section~\ref{sec:strat-aut} generalise to the game setting (Section~\ref{sec:strat-games}). (A full account of the game case is deferred to an extended version of the paper.)
\end{enumerate}
We conclude by summarising our results and discussing future work (Section~\ref{sec:concl}).

\section{Preliminaries}
\label{sec:prel}

\paragraph{Semirings and Monads}
We take as parameter a \emph{commutative semiring} $(S,+,0,\bullet,1)$. Such a semiring induces a preorder $\sqsubseteq$ on $S$ given by $x \sqsubseteq y$ iff there exists $z \in S$ s.t.~$x + z =y$, with $0$ as bottom element. 

\begin{ass}
\label{ass-lattice}
We assume that $(S,\sqsubseteq)$ is a complete lattice with $1$ as top element, and that both $+$ and $\bullet$ preserve suprema of increasing chains and infima of decreasing chains, in each argument.
\end{ass}

\begin{example}
\label{example-semirings}
The following examples of semirings satisfy our assumptions:
\begin{enumerate}[noitemsep]
\item the \emph{boolean semiring} $(\{0,1\},\vee,0,\wedge,1)$, with associated order $\le$ on $\{0,1\}$,
\item the \emph{tropical semiring} ${\mathbb N}^\infty = (\mathbb N^\infty,\min,\infty,+,0)$, with order $\ge$ on ${\mathbb N}^\infty = \mathbb N \cup \{\infty\}$,
\item for $B \in \mathbb N$, the \emph{$B$-bounded variant of the tropical semiring} ${\mathbb N}_B^\infty = ({\mathbb N}_B^\infty,\min,\infty,+_B,0)$, where ${\mathbb N}_B^\infty = \{ n \in \mathbb N \mid 0 \le n \le B\} \cup \{\infty\}$ and where 
$m +_B n = \begin{cases} m + n, \text{ if } m + n \le B\\
\infty, \text{ otherwise} \end{cases}$,
with associated order $\ge$,
\item the \emph{tropical positive rationals} ${\mathbb Q}_{\ge 0}^\infty = (\mathbb Q^\infty,\min,\infty,+,0)$ and its \emph{$B$-bounded variant} for $B \in \mathbb Q^+$, with ${\mathbb Q}_{\ge 0}^\infty = \{ q \in Q \mid q \ge 0 \} \cup \{\infty\}$ and structure similar to the tropical semiring.
\end{enumerate}
\end{example}
A commutative semiring $S$ as above induces a \emph{semiring monad} $(\T_S,\eta,\mu)$, with $\T_S(X) = \textstyle\{\, \varphi : X \to S \mid \supp(\varphi) \text{ is finite}\,\}$, and \emph{unit} and \emph{multiplication} given by
\begin{eqnarray*}
\eta_X(x)(y) = \begin{cases}1 & \text{if } y = x \\
0 & \text{otherwise} \end{cases}, \quad \quad
\mu_X(\Phi)(x) = \textstyle\sum\limits_{\varphi \in \supp(\Phi)}\Phi(\varphi) \bullet \varphi(x)
\end{eqnarray*}
where $\supp(\varphi) = \{ x \in X \mid \varphi(x) \ne 0\}$ is the \emph{support} of $\varphi$. 
We use formal sum notation $\sum\limits_{i \in I} c_i x_i$ with $I$ finite to denote the element of $\T_S (X)$ mapping $x_i$ to $c_i$ for $i \in I$, and $x \not\in \{x_i \mid i \in I\}$ to $0$. We also note that $\T_S (1) \simeq S$, where $1$ denotes a one-element set. In what follows we use $\T_S (1)$ and $S$ interchangeably.

\paragraph{Equational Systems} 
Given a complete lattice $(L,\sqsubseteq)$, a \emph{nested equational system} (cf.~\cite{Niwinski2001}) has the form:\\
\begin{minipage}{0.3\textwidth}
\begin{align*}
\begin{bmatrix}
u_1 & =_\mu & f_1(u_1,\ldots,u_n)\\
u_2 & =_\nu & f_2(u_1,\ldots,u_n)\\
& \vdots \\
u_n & =_\eta & f_n(u_1,\ldots,u_n)
\end{bmatrix}
\end{align*}\\[-15pt]~
\end{minipage}
\begin{minipage}{0.7\textwidth}
~\\[-4pt]with $\eta = \mu$ ($\eta = \nu$) if $n$ is odd (resp.~even), with $u_1,\ldots,u_n$ ranging over $L$. The reader is referred to \cite[Section~1.4.4]{Niwinski2001} for a formal definition of the solution of a nested equational system. Here we only give an informal description. We assume that the most significant equation is the \emph{last} one.
\end{minipage}

\noindent The solution is obtained as follows: 
a solution for $u_1$ parameterised by values for $u_2, \ldots,u_n$ is obtained from the first equation and substituted into the remaining ones; 
a solution for $u_2$ parameterised by values for $u_3,\ldots,u_n$ is obtained from the second equation and substituted into the remaining ones; and so on until  
a closed solution for $u_n$ is obtained; 
this is substituted back to give closed solutions for $u_{n-1}, \ldots, u_1$.

\paragraph{Two-Player Games}
A \emph{game graph} is given by a pair $(Q,E)$ with $Q$ a finite set of \emph{states} partitioned as $Q = Q_\exists \sqcup Q_\forall$, and $E \subseteq Q \times Q$ a relation describing possible \emph{moves}. We write $E_\exists := \{(q,q') \in E \mid q \in Q_\exists\}$ for the set of $\exists$ moves and similarly for $E_\forall$. A \emph{two-player game} played on a game graph $(Q,E)$ by players $\exists$ and $\forall$ starts in an \emph{initial state} $q_0 \in Q$ and proceeds by either $\exists$ or $\forall$ choosing a successor state, depending on whether the current state belongs to $Q_\exists$ or $Q_\forall$. A \emph{(complete) play} is a finite or infinite path $q_0 q_1 \ldots$ with $(q_i,q_{i+1}) \in E$ for all $i$, where if the path ends in $q_n$ then there is no $q \in Q$ with $(q_n,q) \in E$. A \emph{parity game} is a game graph $(Q,E)$ together with a parity map $\Omega : Q \to \mathbb N$ with finite codomain. A complete play in a parity game is won by $\exists$ iff it is either a finite play $q_0 q_1 \ldots q_n$ with $q_n \in \forall$, or an infinite play $q_0 q_1 \ldots$ with $\max \{n \mid n = \Omega(q_i) \text{ for infinitely many } i \}$ even; otherwise the play is won by $\forall$.

A \emph{strategy} for $\exists$ in such a game maps partial plays $q_0 q_1 \ldots q_i$ with $q_i \in Q_\exists$ to a next move $(q_i,q) \in E$. Strategies for $\forall$ are defined similarly. A strategy is \emph{winning for $\exists$} iff $\exists$ wins all plays which conform to this strategy. A strategy for $\exists$ is said to be \emph{finite memory} if it can be encoded using a finite set $M$ along with two update functions, $\sigma_s : M \times Q_\exists \to Q$ and $\sigma_m : M \times (Q_\exists + E_\forall) \to M$, with $(q,\sigma_s(m,q)) \in E$. Assuming an initial memory value $m_0 \in M$, such a pair $(\sigma_s,\sigma_m)$ yield a strategy that in state $q \in Q_\exists$ with memory value $m$ prescribes a move to $\sigma_s(m,q)$; moreover, following the prescribed $\exists$ move from $q$ (an arbitrary $\forall$ move $e \in E$), the memory is updated to $\sigma_m(m,q)$ (respectively $\sigma_m(m,e)$).

\paragraph{Coalgebras}
A \emph{coalgebra} for a functor $F : \Set \to \Set$ is given by a pair $(C,\gamma)$ with $C$ a set (of \emph{states}) and $\gamma : C \to F (C)$ a function (the \emph{transition map}). A \emph{pointed coalgebra} $(C,\gamma,c_0)$ additionally provides an initial state $c_0 \in C$.
The intuition is that, for a state $c \in C$, $\gamma(c)$ describes the one-step observations that one can make of $c$, structured according to $F$ (see \cite[Chapters 1, 2]{JacobsBook} for an introduction to coalgebras).
\begin{example}
Let $\T_S : \Set \to \Set$ be as in 
Section~\ref{sec:prel}. A $\T_S$-coalgebra $(Q,\gamma)$ is the same as a finitely-branching transition system with set of states $Q$ and transitions labelled by quantities in $S$. Given such a quantitative transition system, the corresponding $\gamma$ is given by $\gamma(q)(q') = s$ whenever the transition $q \rightarrow q'$ carries the quantity $s$, and $\gamma(q)(q') = 0_S$ whenever there is not transition $q \rightarrow q'$.
\end{example}
\cite{CirsteaSH17} consider coalgebras of type $\T_S \circ F$, where $\T_S : \Set \to \Set$ is as above 
and $F : \Set \to \Set$ is given by $F (X) = \coprod_{\lambda \in \Lambda} X^{\arity(\lambda)}$, with $\coprod$ denoting disjoint union and $\Lambda$ a set of operation symbols with finite arities. Thus, an $F$-coalgebra $(C,\gamma)$ (used in what follows to model \emph{runs} of a parity automaton) consists of a set of states $C$ together with, for each state $c \in C$, a transition of type $c \to (\lambda,c_1,\ldots,c_{\arity(\lambda)})$ for some $\lambda \in \Lambda$. Taking $F (X) = A \times X \simeq \coprod_{a \in A} X$ will allow us to cover (automata over) infinite words over $A$, while taking $F (X) = A \times X \times X \simeq \coprod_{a \in A} X \times X$ will cover infinite trees with nodes labelled by elements of $A$. More complex choices of $F$ are also possible: e.g.~$F (X) = X + X\times X + \{*\}$ models runs which combine word-like structure with tree-like structure as well as termination: in each state of the run, there is a \emph{choice} of continuing in another state (first term), splitting into two states (second term) or terminating (last term). Finally, the presence of $\T_S$ in the coalgebra type will model the existence of \emph{several} transitions from each state of an automaton, with each transition being assigned a quantity in $S$.

\section{Weighted Parity Automata with Offsetting}
\label{sec:wpao}

Throughout this section, $(S,+,0,\bullet,1)$ is a commutative semiring satisfying the assumptions in 
Section~\ref{sec:prel}. We also assume $F : \Set \to \Set$ is given by $F(X) = \coprod_{\lambda \in \Lambda} X^{\arity(\lambda)}$ with $\Lambda$ a set of operation symbols with finite arities. The functor $F$ specifies the structure of individual runs of a quantitative automaton. While an elementary definition of weighted parity automata is possible for each $F$, the following coalgebraic definition, parameterised by $F$, captures all such instances.
\begin{definition}[\cite{CirsteaLMCS}]
\label{def:wpao}
A \emph{weighted parity automaton with offsetting (WPAO)} is given by a $S \times (\T_S \circ F)$-coalgebra $(Q,\langle r,\gamma \rangle)$ together with a \emph{parity map} $\Omega : Q \to \mathbb N$ assumed to have finite range. We call $r : Q \to S$ the \emph{offset map} and $\gamma : Q \to \T_S (F (Q))$ the \emph{transition map}. 
For $n \in \ran(\Omega)$, we write $Q_n = \{q \in Q \mid \Omega(q) = n\}$. If $\ran(\Omega) = \{1,2\}$, we call $\A$ a \emph{weighted B\"uchi automaton with offsetting (WBAO)}.
\end{definition}
WPAOs generalise a standard parity automata in several ways: transitions carry weights from $S$, states carry offsets, also from $S$, and the structure of transitions strictly subsumes automata over words/trees.

Taking $F X = \Id$ (the identity functor), $S = (\{0,1\},\vee,0,\wedge,1)$ and ignoring offsets yields the standard notion of parity automaton with the alphabet a singleton set. Taking $S = (\mathbb N^\infty,\min,\infty,+,0)$ and $F = \Id$ yields a variant of parity automata over infinite words where both transitions and states are assigned values in $S$, to be viewed as resources consumed by transitions, and respectively resources gained upon visiting states. 
In the latter case, we will use an offsetting operation $\varominus : \mathbb N^\infty \times \mathbb N^\infty \to \mathbb N^\infty$ to capture the fact that resources gained can only be used for future computation. This operation is defined by 
$n \varominus m = \begin{cases}
\max(n - m,0), ~\text{ if } m \ne \infty \text{ or } n \ne \infty\\
\infty,~\text{ otherwise}
\end{cases}$, 
with $\max$ the standard one on $\mathbb N^\infty$. Thus, $\varominus$ is a subtraction operation capped below at $0 \in \mathbb N$; this matches the intuition that $m$ acts as an offset value.

The operation $\varominus$ on $\mathbb N^\infty$ generalises to an arbitrary, and even \emph{partial} semiring $(S,+,0,\bullet,1)$ \cite{CirsteaLMCS}. Specifically, the binary operation $\varoslash : S \times S \to S$ is defined by 
$s \varoslash t = \inf \{ u \mid u \bullet t \sqsupseteq s \}\,$. 
The notation used reflects that $\varoslash$ is \emph{almost} an inverse to the semiring multiplication. Well-definedness follows by Assumption~\ref{ass-lattice}. When $S = (\{0,1\},\vee,0,\wedge,1)$, $\varoslash: S \times S \to S$ instantiates to the first projection.

\begin{proposition}[\cite{CirsteaLMCS}]
\label{varoslash-prop}
The following hold for any 
$(S,+,0,\bullet,1)$ satisfying the assumptions in 
Section~\ref{sec:prel}:
 \begin{enumerate}[noitemsep]
 \item $(a \varoslash c) + (b \varoslash c) = (a + b) \varoslash c$,
 \item $\sup_{i \in \omega}(a_i \varoslash c) = (\sup_{i \in \omega} a_i) \varoslash c$, for $a_0 \sqsubseteq a_1 \sqsubseteq \ldots$,
 \item $\inf_{i \in \omega}(a_i \varoslash c) = (\inf_{i \in \omega} a_i) \varoslash c$, for $a_0 \sqsupseteq a_1 \sqsupseteq \ldots$.
 \end{enumerate}
 \end{proposition}
Next, we define the notion of an \emph{(accepting) run} of a WPAO (which generalises the standard notions from parity word/tree automata), and use transition weights and state offsets to associate \emph{values} to runs.
\begin{definition}
\label{run}
\label{def:accepting-run}
Given a set $Q$ (of states), an \emph{$F$-run with states in $Q$} is a possibly infinite tree with nodes labelled by pairs $(q,\lambda)$ with $q \in Q$ and $\lambda \in \Lambda$, and with each node $(q,\lambda)$ having exactly $\arity(\lambda)$ children $(q_1,\lambda_1), \ldots,(q_n,\lambda_n)$. For a WPAO $\A = (Q,\langle r,\gamma \rangle)$, a \emph{run of $\A$} additionally satisfies $\gamma(q)(\iota_\lambda(q_1,\ldots,q_n)) \ne 0_S$ for each node $(q,\lambda)$ with children $(q_1,\lambda_1), \ldots,(q_n,\lambda_n)$\footnote{That is, a transition from $q$ to $(\lambda,q_1,\ldots,q_n)$ exists in $\A$.}. We write $Runs$ for the set of $F$-runs with states in $Q$, and $\Runs_\A$ for the set of runs of $\A$. 
A run is \emph{accepting} iff for each branch of the associated tree, the maximum parity occurring infinitely often is even. 
\end{definition}
Thus, at each state $q$, a run of $\A$ selects an available transition, say of type $\lambda$, and proceeds with a number of successors equal to $\arity(\lambda)$ -- this can be $0$ (termination), $1$ (linear structure) or $\ge 2$ (tree structure).

\begin{definition}
\label{def:value-run}
For a WPAO $\A = (Q,\langle r,\gamma \rangle)$, the function $\val : Runs \to S$ is the greatest fixpoint of the operator on $S^{Runs}$ taking $v : Runs \to S$ to $v': Runs \to S$, where for $z \in Runs$ with root $(q,\lambda)$ and children $z_1,\ldots,z_n$ labelled by $q_1, \ldots,q_n$, respectively:
\[v'(z) = (\gamma(q,\iota_\lambda(q_1,\ldots,q_n)) \bullet v(z_1) \bullet \ldots \bullet v(z_n)) \varoslash r(q) \]
\end{definition}
The above coinductive definition formalises the idea that, in order to obtain the value of a run $z$ with root $(q,\lambda)$ and children $z_1,\ldots,z_n$, one computes increasingly finer approximations for $\val(z)$ by using, at each step, the current approximations for $\val(z_1), \ldots,\val(z_n)$; at each state $q$ along the run, the value $r(q)$ is used to offset the "future" value. As a result, the function $\val$ maps any $z \in Runs$ which is not a run of $\A$ to $0_S$. The existence of the required fixpoint follows by Kleene's theorem, from $(S,\sqsubseteq)$ being a complete lattice together with the operator involved being co-continuous (by Assumption~\ref{ass-lattice} and Proposition~\ref{varoslash-prop}).


\begin{example}
When $S = (\mathbb N^\infty,\min,\infty,+,0)$, the value of a run measures the minimal resources required in the initial state to ensure that the run can be executed (assuming that resources are increased using offset values and decreased using transition costs). The use of $\bullet$ in Definition~\ref{def:value-run} ($+$ on $\mathbb N^\infty$) amounts to adding the costs of all transitions across all branches of a run, whereas the use of $\varoslash$ ($\varominus$ on $\mathbb N^\infty$) ensures that the resource gain in each visited state offsets the cost of \emph{future} transitions. When $S = (\{0,1\},\vee,0,\wedge,1)$ and assuming trivial offset values of $1$, a run has value $1$ iff it is a run of $\A$ in the standard sense. 
\end{example}
We now recall the notion of \emph{extent} of a WPAO, which generalises non-emptiness of parity automata \cite{CirsteaLMCS}. The extent associates a value in S to each automaton state, with the intuition that this value "amalgamates" (via the semiring addition) the values of \emph{all} accepting runs from that state. Crucially, this amalgamation is performed in a step-wise fashion and does not require the values of individual runs to be known. The coalgebraic formulation, if somewhat abstract, allows for a generic and concise definition.

\begin{definition}[Extent of WPAO \cite{CirsteaLMCS}]
\label{extent-def}
Let $\A = (Q,\langle r,\gamma \rangle,\Omega)$ be a WPAO with $\ran(\Omega) = \{1,\ldots,n\}$. For $k \in \ran(\Omega)$, let $\gamma_k : Q_k \to \T_S (F (Q))$ and $r_k : Q_k \to S$ denote the restrictions of $\gamma$ and respectively $r$ to $Q_k := \{q \in Q \mid \Omega(q) = k\}$. The \emph{extent} $\ext_\A = [e_1,\ldots,e_n] : Q \to S$ of $\A$ is the solution of the following nested equational system, with the most significant equation being the last one:
\begin{align}
\label{eqn-extent}\begin{bmatrix}
u_1 & =_\mu  & (\mu_1 \circ T_S (\bullet_F) \circ \T_S F [u_1, \ldots, u_n] \circ \gamma_1) \varoslash r_1\\
u_2 & =_\nu  & (\mu_1 \circ T_S (\bullet_F) \circ \T_S F [u_1, \ldots, u_n] \circ \gamma_2) \varoslash r_2\\
& \vdots \\
u_n & =_\eta & (\mu_1 \circ T_S (\bullet_F) \circ \T_S F [u_1, \ldots, u_n] \circ \gamma_n) \varoslash r_n\end{bmatrix}
\end{align}
with $\eta = \mu$ ($\eta = \nu$) if $n$ is odd (resp.~even), with variables $u_k$ ranging over the poset $(S^{Q_k},\sqsubseteq)$ (and thus $[u_1,\ldots,u_n] : Q \to S$), and with the first operands in the rhs$s$ pictured below:
\begin{align*}
\UseComputerModernTips\xymatrix@-0.5pc{Q_k \ar[r]^-{\gamma_k} & \T_S F Q \ar[rrr]^-{\T_S F [u_1, \ldots,u_n]} & & & \T_S F S \ar[rr]^-{\T_S (\bullet_F)} & & \T_S S = \T_S^2 1 \ar[r]^-{\mu_1} & \T_S 1 = S}
\end{align*}
In the above, $\bullet_F : F S \to S$ is given by $\bullet_F(\iota_\lambda(s_1,\ldots,s_{\arity(\lambda)})) = s_1 \bullet \ldots \bullet s_{\arity(\lambda)}$ for $\lambda \in \lambda$, and $\mu_1 : \T_S S = \T_S \T_S 1 \to \T_S 1 = S$ is the monad multiplication.
\end{definition}
Take $S = (\mathbb N^\infty, \min,\infty,+,0)$ in Definition~\ref{extent-def}. By (\ref{eqn-extent}), the extent value associated to $q \in Q_k$ is the minimum (use of monad multiplication $\mu_1$) across all transitions from $q$ (use of $\gamma_k$) of the combined extents (use of $\bullet$, in this case $+$) of the successors of $q$ in that transition, and where offsets (use of $r_k$) are applied to account for resource gain in $q$. 
The existence of the nested fixpoint follows again by Kleene's theorem.

\begin{example}
\label{ex:simple}
Consider the resource-aware automaton on the left below ($F = \Id$), with offsets of $0$ in $x$ and $y$, $2$ in $y_1$, and $4$ in $y_2$. The associated equational system, on the right below, has solution $(1,1,0,0)$:

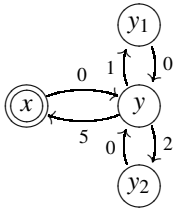
\begin{minipage}{0.5\textwidth}
{\small \qquad $\entrymodifiers={++[o][F-]}\UseComputerModernTips\xymatrix@-1pc{
*{} & *{} & y_1 \ar@/^0.5pc/[d]^-{0} \\
*++[o][F=]{x} \ar@/^0.5pc/[rr]^-{0} & *{} & y \ar@/^0.5pc/[u]^-{1} \ar@/^0.5pc/[ll]^-{5} \ar@/^0.5pc/[d]^-{2} \\
*{} & *{} & y_2 \ar@/^0.5pc/[u]^-{0}
}$}
\end{minipage}
\begin{minipage}{0.5\textwidth}
$\begin{bmatrix}
e_x & =_\nu & e_y\\
e_y & =_\mu & \min(e_x+5,e_{y_1}+1,e_{y_2}+2)\\
e_{y_1} & =_\mu & e_y \varominus 2\\
e_{y_2} & =_\mu & e_y \varominus 4\end{bmatrix}$
\end{minipage}

\noindent This matches the intuition that an initial resource value of $1$ suffices in both $x$ and $y$ to visit the accepting state $x$ infinitely often, while never running out of resources: the transitions $x \rightarrow y$ and $y \rightarrow y_1$, together consuming $1$ resource, can be  used to reach $y_1$, and no additional resources are required from $y_1$ or $y_2$, given the resource gains in these states: one can loop through $y_1$, $y$ and $y_2$ until the available resources in $y$ reach $6$ and the transition to $x$ can be taken. On the other hand, the run $x\, y\, x\, y\, \ldots$ has value $\infty$.
\end{example}
As expected, no \emph{accepting} run of a WPAO can have a value that is strictly above (w.r.t.~$\sqsubseteq$) the extent of its initial state. The proof of this result (omitted due to space limitations) uses coinduction. 
\begin{proposition}
\label{prop-ext}
Let $\A = (Q,\langle r,\gamma \rangle,\Omega)$ be a WPAO, and let $z_0 \in Runs_\A$ be an accepting run of $\A$ with root $q_0$. Then, $\ext_\A(q_0) \sqsupseteq \val(z_0)$. 
\end{proposition}

When $S = (\{0,1\},\vee,0,\wedge,1)$ and with a trivial offset function, the extent associates a value of $1$ to an automaton state iff there exists an accepting run from that state. Also, when $S = (\mathbb N^\infty, \min,\infty,+,0)$, the extent associates to a state $q$ the minimum value of a run (Definition~\ref{run}), calculated across all accepting runs from $q$: that the extent is $\le$ the minimum value follows from Proposition~\ref{prop-ext}; the existence of a run with the same value as the extent will follow from Theorem~\ref{thm1}.

\cite{CirsteaSH17} describes an algorithm for computing extents of weighted parity automata (no offsetting), under the additional assumption that the order $\sqsubseteq$ admits no infinite strictly ascending/descending chains. The algorithm, specialised to the resource-aware setting ($B$-bounded variant of the tropical semiring, which satisfies the above assumption), is given in Figure~\ref{fig1}. If $N_i = |Q_i|$ and $E_i$ is the number of edges from states in $Q_i$ (where an edge of type $\lambda$ is counted $\arity(\lambda)$ times), for $i\in \ran(\Omega)$, and if $k = |\ran(\Omega)|$, the algorithm has complexity $O(N_k \times (N_k + E_k + N_{k-1} \times (N_{k-1} + E_{k-1} + \ldots + N_1 \times (N_1 + E_1))))$. Thus, if $n = \sum_{i \in \ran(\Omega)} (N_i + E_i)$ is the size of the automaton, the complexity is $O(n^{k+1})$. The actual upper bound for the time complexity is in fact proportional to $B^{|\ran(\Omega)|}$, since for each $i \in \ran(\Omega)$, the {\small \bfseries repeat \ldots until} statement in $Extent(i)$ is executed $B \times N_i$ times in the worst case.

~\\[-35pt]
\begin{figure}[H]
\caption{Algorithm for computing extents of resource-aware WPAOs \cite{CirsteaSH17}}
~\\[-5pt]
      \begin{minipage}{0.6\textwidth}
      \small \noindent\fbox{
      \parbox{96mm}{
  {\bfseries Input:} resource-aware WPAO $ \A = (Q,\langle r,\gamma \rangle,\Omega)$

  {\bfseries Output:} extent $e : Q \to \mathbb N \cup \{\infty\}$ of $\A$\\[-18pt]
  \begin{enumerate}
  \itemsep=-3pt
  \leftskip=1mm
  \item $Extent(\max(\ran(\Omega)))$
  \end{enumerate}~\\[-12pt]
  Procedure $Extent(n \in \mathbb N)$  ~\\[-18pt]
  \begin{enumerate}
  \itemsep=-3pt
  \leftskip=1mm
  \item {\bfseries if} $n = 0$ {\bfseries then return} {\bfseries endif}
  \item {\bfseries for} $q \in Q_n$ {\bfseries do}~\\[-12pt]
  \item ~ ~ {\bfseries let} $e(q) := \begin{cases} 0, \text{ if } n \text{ is even} \\ \infty, \text{ if } n \text{ is odd} \end{cases}$~\\[-10pt]
  \item {\bfseries endfor}
   \item {\bfseries repeat}
  \item ~ ~ {\bfseries let} $old := e$
  \item ~ ~ e := $Extent(n-1)$
  \item ~ ~ {\bfseries for} $q \in Q_n$ {\bfseries do}
  \item ~ ~ ~ ~ $e(q) \leftarrow \min_{i \in I} ((v_i + old(q_1) + \ldots + old(q_{\arity(\lambda)})) \ominus r(q))$
  \item[] ~ ~ ~ ~ ~ ~ {\bfseries where} $\gamma(q) = \sum_{i \in I}v_i \,\iota_\lambda(q_1,\ldots,q_{\arity(\lambda)})$
   \item ~ ~ {\bfseries endfor}
  \item {\bfseries until} $e = old$
  \end{enumerate}
}
}
  \end{minipage}
\begin{minipage}{0.4\textwidth}
  \begin{itemize}[noitemsep]
  \item Lines 8--10 of the $Extent$ procedure compute a better approximation of the extent for automaton states with the current parity $n$, based on a one-step unfolding of the transition structure.
  \item Line 7 computes the extent of states with immediately lower parity, relative to the current values for states with parity $n$, through a recursive call to $Extent$ (which may involve further recursive calls).
  \item Recursive calls to $Extent$ update the same copy of $e$, and only make an additional copy to remember values from the previous step.
  \end{itemize}
\end{minipage}

  \label{fig1}
  \end{figure}

\section{Strategies in Resource-Aware B\"uchi Automata}
\label{sec:strat-aut}

For the remainder of the paper we focus on \emph{resource-aware WBAOs}, that is, we fix $S = ({\mathbb N}_B^\infty,\min,\infty,+,0)$ with $B \in \mathbb N$ and assume $\ran(\Omega) = \{1,2\}$. While some of our basic results apply more generally to resource-aware \emph{parity} automata, extensions of our main results in this section (Theorems~\ref{thm} and \ref{thm2}) to parity automata are only sketched, with a full account being left as future work.

For the above choice of semiring, the generic order $\sqsubseteq$ instantiates to the $\ge$ order on ${\mathbb N}_B^\infty$. When referring to least/greatest fixpoints, we will implicitly assume the generic order $\sqsubseteq$. Yet, depending on context, we will also talk about least upper bounds/greatest lower bounds w.r.t.~the order $\le$ on ${\mathbb N}_B^\infty$.

Our results concern the construction of strategies for exhibiting \emph{optimal} runs.
\begin{definition}
A run from $q_0$ in a WPAO $\A$ is \emph{optimal} if it is accepting and has value $\ext_\A(q_0)$.
\end{definition}
Unlike standard parity automata on infinite words/trees, strategies that exhibit optimal runs in resource-aware automata are not memoryless; intuitively, this is because proceeding from a non-accepting state $q$ with available resources equal to $\ext_\A(q)$ may require accumulating additional resources in $q$ before moving towards an accepting state. We tackle the problem of finding strategies for exhibiting optimal runs in a WPAO by moving to a standard parity game called the \emph{resource-aware game} (Section~\ref{sec:res-model}). In this game, available resources are recorded within states, $\exists$ moves correspond to transitions that can be taken with the currently available resources, and $\forall$ moves arise from the presence of branching in runs (e.g.~in the case of WBAOs on trees). We show here how the algorithm in Figure~\ref{fig1} can be extended to compute a memoryless winning strategy for $\exists$ in the resource game, and use this to derive a memory-full strategy for exhibiting an optimal run in $\A$ (Section~\ref{sec:memory-full}). We then show that in fact simpler strategies, combining a \emph{good-for-energy} strategy with an \emph{attractor} strategy, exist 
(Section~\ref{sec:reduced-mem-strat}).

\subsection{Resource Games for WPAOs}
\label{sec:res-model}

\begin{definition}[Resource Game]
\label{res-model-def}
Let $\A = (Q,\langle r,\gamma \rangle, \Omega)$ be a resource-aware WPAO. The \emph{resource game $\R_\A$ of $\A$} is a standard parity game with states $Q_\exists = \{(q,n) \in Q \times \mathbb N \mid n \ge \ext_{\A}(q)\}$ and $Q_\forall = F (Q_\exists)$, parities inherited from $\A$ on $\exists$ states and equal to $\min \{\Omega(q) \mid q \in Q\}$ on $\forall$ states, and moves given by:
\begin{eqnarray*} 
(q,n) \UseComputerModernTips\xymatrix@-0.5pc{\ar[r]^-{\exists} &} (\lambda,(q_1,n_1),\ldots,(q_m,n_m)) & \text{whenever}~ 
n + r(q) \ge \gamma(q)(\iota_\lambda(q_1,\ldots,q_m)) + n_1 + \ldots + n_m\,,\\
(\lambda,(q_1,n_1),\ldots,(q_m,n_m)) \UseComputerModernTips\xymatrix@-0.5pc{\ar[r]^-{\forall} &} (q_i,n_i) & \text{for all}~ i \in \{1,\ldots,m\}\,. \qquad \qquad \qquad \qquad \qquad \qquad \qquad 
\end{eqnarray*}
\end{definition}
The reason for requiring $n \ge \ext_\A(q)$ in the definition of $Q_\exists$ is that, by Proposition~\ref{prop-ext}, at least $\ext_\A(q)$ resources are required to exhibit an \emph{accepting} run from $q$. A transition from $q$ to $(\lambda,q_1,\ldots,q_m)$ in $\A$ with weight $w = \gamma(q)(\iota_\lambda(q_1,\ldots,q_m))$ gives rise to one $\exists$ move from $(q,n)$ in $\R_\A$ for each choice to distribute the available resources following the transition (that is, after gaining $r(q)$ and consuming $w$) to acceptable resources for the states $q_1,\ldots,q_m$ (that is, $n_i \ge \ext_\A(q_i)$ for $i \in \{1,\ldots,m\}$). The existence of at least one such move is guaranteed by the definition of extents together with the assumption that extents are finite for all $\exists$ states in $\R_\A$. Finally, note that assigning the lowest parity to $\forall$ states results in these parities not affecting the winner of a play, since $Q_\exists$ and $Q_\forall$ states alternate strictly in any play.

\begin{example}
\label{ex:simple-full}
Part of the game model associated to the WBAO in Example~\ref{ex:simple} is depicted below. Since $F = \Id$, $\forall$ moves are fully determined and therefore not shown.
{\small \begin{eqnarray*}
\entrymodifiers={++[o][F-]}\UseComputerModernTips
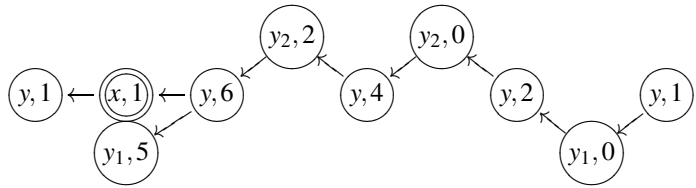
\xymatrix@-1.9pc{
*{} & *{} & *{} & *{} & *{} & *{} & *{} & *{} & *{} & *{} & *{} & *{} & *{} & y_2,2 \ar[dl] & *{} & y_2,0 \ar[dl]  & *{} & *{} & *{} & *{} & *{} & *{} & *{} \\ 
y,1 & *{} & *{} & *{} & *{} & *{} & *++[o][F=]{x,1}\ar[llllll] & *{} & *{} & *{} & *{} & *{} & y,6 \ar[llllll] \ar[dllllll] & *{}  & y,4 \ar[ul] & *{} & y,2 \ar[ul] & *{} & y,1 \ar[dl] \\
*{} & *{} & *{} & *{} & *{} & *{} & y_1,5 & *{} & *{} & *{} & *{} & *{} & *{} & *{} & *{} & *{} & *{} & y_1, 0 \ar[ul]
}
\end{eqnarray*}}
\!\!Take, for example, the two transitions from state $(y,6)$. The transition leading to $(x,1)$ consumes the $5$ resources required by $y \to x$, whereas the transition to $(y_1,5)$ consumes a single resource as required by $y \to y_1$. On the other hand, the transition from $(y_2,2)$ to $(y,6)$ consumes no resources, while additionally gaining $4$ resources as a result of visiting state $y_2$.
\end{example}

Next, we define runs of $\A$ that never run out or resources.
\begin{definition}
\label{res-aware-run}
A \emph{resource-aware run} of a WPAO $\A$ is a run of $\A$ (Definition~\ref{run})  whose states can be annotated by values in $\mathbb N$ in such a way that paths through the annotated tree correspond to plays in $\R_\A$.
\end{definition}
The characterisation of the value of a run as a greatest fixpoint w.r.t.~$\sqsubseteq$ (least fixpoint w.r.t.~$\le$) can be used to show that the value of a resource-aware run with root $q_0$ annotated by $r_0$ is at most $r_0$.

\begin{proposition}
\label{lemma1}
If $z$ is a resource-aware run of a WPAO $\A$, with root annotated by $r_0$, then $\val(z) \le r_0$.
\end{proposition}
By Propositions~\ref{prop-ext} and \ref{lemma1}, an \emph{accepting} run can only exist from state $(q,r)$ in $\R_\A$ if $r \ge \ext_\A(q)$. To show that an accepting run from $(q,\ext_\A(q))$ does indeed exist, we will construct a strategy to exhibit it.

\subsection{Memory-full Strategies in WBAOs}
\label{sec:memory-full}

From now on we assume, for the most part, that $\A$ is a \emph{WBAO}. The extension to parity automata is sketched at the end of Section~\ref{sec:reduced-mem-strat}. We first consider a sub-game of $\R_\A$ derived from the execution of the algorithm in Figure~\ref{fig1}. Specifically, we assume that the values $e(q)$ with $q \in Q_2$ are already computed, and consider the last recursive call to $Extent(1)$ and the updates made to $e(q)$ with $q \in Q_1$ within this call. An assignment to $e(q)$ qualifies as an update if the value of $e(q)$ is \emph{strictly} decreased.
\begin{definition}
Let $\A = (Q,\langle r,\gamma \rangle, \Omega)$ be a WBAO. The \emph{resource sub-game $\R'_\A$} of $\R_\A$ has states
\begin{eqnarray*}
Q'_\exists = \{(q,n) \in Q_1 \times \mathbb N \mid \text{ an update $e(q) \leftarrow n$ is performed} \,\} \cup \{(q,\ext_\A(q)) \mid q \in Q_2\} ,~ Q'_\forall = F Q'_\exists
\end{eqnarray*}
and moves inherited from $\R_\A$.
\end{definition}
As a result, $\R'_\A$ includes all the moves of $\R_\A$ which witness updates to $e(q)$, with $q \in Q_1$, as well as moves from $(q,\ext_\A(q))$ which witness the value of $(q,\ext_\A(q))$, with $q \in Q_2$. Clearly, $Q'_\exists$ is finite and contains all states $(q,\ext_\A(q))$ with $q \in Q$ and $\ext_\A(q)$ finite. Moreover, we have:
\begin{proposition}
\label{prop1}
Let $\A = (Q,\langle r,\gamma \rangle, \Omega)$ be a WBAO. Then, $\exists$ has a memoryless winning strategy from any state $(q,\ext_\A(q))$ in $\R'_\A$ (and hence also from any state $(q,\ext_\A(q))$ in $\R_\A$).
\end{proposition}
The proof (omitted) constructs a memoryless strategy $\sigma$ for $\exists$ in $\R'_\A$ that in accepting states $(q,\ext_\A(q))$ proceeds according to a transition witnessing the value of $\ext_\A(q)$ -- such a transition is always of the form $(q,\ext_\A(q)) \rightarrow (\lambda,(q_1,\ext_\A(q_1)),\ldots,(q_{\arity(\lambda)},\ext_\A(q_{\arity(\lambda)})))$, and in non-accepting states $(q,n)$ selects a transition witnessing the update of $e(q)$ to $n$. The resulting strategy provides a way to either terminate or reach an accepting state $(x,\ext_\A(x))$ from any state $(q,\ext_\A(q))$ in $\R'_\A$ via one or more transitions. Moreover, by construction $\sigma$ has \emph{no redundancy}, in that a $\sigma$-conform play can not visit a non-accepting state $(y,n)$ and subsequently a state $(y,n')$ with $n' < n$ without visiting an accepting state in-between. This property, together with the finiteness of $\sigma$, are captured by the following definition.

\begin{definition}
A \emph{skeleton strategy} for a WPAO $\A$ is a strategy for $\exists$ in a finite sub-game $\R$ of $\R_\A$ s.t.:
\begin{itemize}[noitemsep]
\item $\R$ includes all states $(q,\ext_\A(q))$ with $\ext_\A(q) \ne \infty$, the $\exists$ moves in $\R$ are those in $\R_\A$ with source and target in $\R$, and if a $\forall$ state belongs to $\R$ then so are all $\R_\A$-moves from that state;
\item $\sigma$ has no redundancy.
\end{itemize}
\end{definition}
Given the constraint regarding $\forall$ moves, a \emph{winning} skeleton strategy for $\exists$ in $\R$ yields a winning strategy for $\exists$ in $\R_\A$ from $(q,\ext_\A(q))$. Note also that $\R$ can be taken to be $\R'_\A$. We show next that such a winning skeleton strategy induces a memory-full strategy for exhibiting an optimal run of $\A$ from any $q \in Q$ with $\ext_\A(q) \ne \infty$. The induced strategy uses as memory the value of available resources.

\begin{definition}[Memory-full strategy induced by skeleton]
\label{rho}
Let $\sigma$ be a skeleton strategy for a WPAO $\A$. The \emph{memory-full strategy $\rho$ induced by $\sigma$} uses values in $\mathbb N_B$ as memory and is defined by:
\begin{itemize}[noitemsep]
\item in state $q \in Q$ with memory $s \ge \ext_\A(q) \in \mathbb N_B$, if $\sigma(q,s) = (\lambda,(q_1,s_1),\ldots,(q_n,s_n))$ then $\rho$ prescribes parallel moves to each $q_i$, with associated memory value $s_i$, for $i \in \{1,\ldots,n\}$.
\end{itemize}
\end{definition}
Thus, $\rho$ is defined in state $q$ with memory value $s$ precisely when $(q,s)$ belongs to $\R$. As a result, playing $\rho$ from state $q \in Q$ with initial memory value $\ext_\A(q)$ will always lead to states where $\rho$ is defined.

\begin{theorem}
\label{thm1}
Let $\sigma$ and $\rho$ be as in Definition~\ref{rho}. If $\sigma$ is winning for $\exists$ in $\R$, the run from $q \in Q$ obtained by playing $\rho$ (in parallel) with initial memory $\ext_\A(q)$ is accepting and has value $\ext_\A(q)$.
\end{theorem}
\begin{proof}
That any $\rho$-conform run from $q \in Q$ with initial memory $\ext_\A(q)$ is accepting follows immediately from the definition of $\rho$. Also, by Proposition~\ref{lemma1}, the value of an $\R_\A$-run starting in $(a,\ext_\A(a))$ cannot exceed $\ext_\A(a)$. Since by Proposition~\ref{prop-ext} there can not exist \emph{any} accepting run with value strictly below $\ext_\A(a)$, it follows that the value of the run is precisely $\ext_\A(a)$.
\end{proof}
In particular, Theorem~\ref{thm1} applies to the strategy $\sigma$ of Proposition~\ref{prop1}. Moreover, the algorithm in Figure~\ref{fig1} can be extended to compute $\sigma$, once extents have been computed. We only outline the algorithm here, as an improved version is detailed in Section~\ref{sec:reduced-mem-strat}. The algorithm is an enhanced version of the call to $Extent(1)$ in the algorithm in Figure~\ref{fig1}, which assumes that $\ext_\A(x)$ is known for $x \in Q_2$. In addition to (re-)computing $\ext_\A(y)$ for $y \in Q_1$, the algorithm computes a partial function $\sigma : Q_1 \times \mathbb N \pto Q$ with domain $\dom(\sigma) \supseteq \{ (q,s) \in Q_1 \times \mathbb N \mid s = \ext_\A(q)\}$ as follows: every time the value of $e(q)$ with $q \in Q_1$ is updated (line 9), and this is witnessed by an $\A$-transition $q \rightarrow (\lambda,q_1,\ldots,q_{\arity(\lambda)})$, set $\sigma(q,e(q))$ to $(\lambda,q_1,\ldots,q_{\arity(\lambda)})$ and update the memory accordingly (so as to proceed with memory value $old(q_i)$ from $q_i$, for $i \in \{1,\ldots,n\}$). Finally, for $q \in Q_2$, define $\sigma(q,\ext_\A(q))$ based on \emph{any} transition that witnesses the value of $\ext_\A(q)$, and again, update the memory accordingly.

\begin{example}
\label{ex:strat-full}
For the automaton in Example~\ref{ex:simple}, removing the move $(y,6) \to (y_1,5)$ from the moves in Example~\ref{ex:simple-full} yields a winning skeleton strategy.
\end{example}

\subsection{Reduced-Memory Strategies for $\A$}
\label{sec:reduced-mem-strat}

A drawback of the memory-full strategy $\rho$ induced by a skeleton strategy $\sigma$ for $\exists$ is the amount of memory required to store the values triggering different moves for each $q \in Q_1$. For instance, the strategy in Example~\ref{ex:strat-full} prescribes moves to each of $y_1$, $y_2$ and $x$ in state $y$, depending on the current memory value. In general, the number of moves is only bounded by the number of states. We now show that substantially simpler strategies exist. Specifically, we consider strategies wherein only \emph{two} moves are prescribed for each $q \in Q_1$: one aimed at reaching an accepting state (provided sufficient resources are available), and one aimed at increasing the available resources in $q$. We obtain such a strategy by discarding certain moves from a skeleton strategy. However, for the resulting strategy to remain winning, its definition must be modified so as to allow unused resources to be carried over.

To simplify presentation, we only consider \emph{word} automata, i.e.~$\arity(\lambda) \le 1$ for $\lambda \in \Lambda$. In this case, resource-aware accepting runs of $\A$ are essentially plays in $\R_\A$ won by $\exists$ -- this is because $\forall$ moves are fully determined, effectively resulting in a \emph{one-player} game with skeleton strategies having type $Q \times \mathbb N_B \pto Q \times \mathbb N_B$. We explain how our results generalise to tree automata at the end of the section.

The next definition is a variation of Definition~\ref{rho} which allows unused resources to be carried over.
\begin{definition}[Carry-over memory-full strategy induced by skeleton]
\label{def:skeleton-induced}
Let $\sigma$ be a skeleton strategy for $\A$. The \emph{carry-over memory-full strategy $\overline \rho$ induced by $\sigma$} is defined as follows:
\begin{itemize}[noitemsep]
\item in state $q \in Q$ with memory $s \in \mathbb N$, let $n = \max\{\,m \mid m \le s,\, (q,m) \in \dom(\sigma)\}$ and let $\sigma(q,n) = (q',s')$; then $\overline\rho$ prescribes a move to $q'$ and updates the memory to $s' + (s-n)$.
\end{itemize}
\end{definition}
Thus, in state $q$ with memory $s$, $\overline\rho$ proceeds according to the $\sigma$-move in $q$ that requires the \emph{most} resources, $n$, not exceeding $s$. When not all the resources $s$ are required by the move (that is, $s > n$), the unused resources ($s - n$) are carried over. 
We can then prove a result similar to Theorem~\ref{thm1}.
\begin{theorem}
\label{thm}
Let $\A$ be a WBAO. If $\sigma$ and $\overline{\rho}$ are as in Definition~\ref{def:skeleton-induced} and $\sigma$ is winning, then the run from $q \in Q$ obtained by playing $\overline \rho$ with initial memory $\ext_\A(q)$ is accepting and has value $\ext_\A(q)$.
\end{theorem}
\begin{proof}
The only difference between the $\overline{\rho}$ here and the $\rho$ in Theorem~\ref{thm1} is that $\overline{\rho}$ carries over unused resources. This may result in a \emph{different} path to an accepting state $(x,n)$ with $n \ge \ext_\A(x)$, in case the resources carried over result in $\overline{\rho}$ making use of a $\sigma$-move $(q,n) \to (q',n')$ in a situation when $\rho$ would have used a $\sigma$-move $(q,m) \to (q',m')$ with $m < n$. However, the fact that $\sigma$ has no redundancy ensures that any such "jumps" do not result in cycling forever through non-accepting states -- no $\overline{\rho}$-conform play will ever decrease the resources available in a state $q$, without first visiting an accepting state.
\end{proof}
In particular, Theorem~\ref{thm} applies to the skeleton $\sigma$ from Proposition~\ref{prop1}. We show next how to simplify the induced strategy $\overline{\rho}$ while maintaining its winning property.
\begin{definition}[Reduced-memory strategy induced by skeleton]
\label{red-strat}
Let $\sigma$ be a skeleton strategy for $\A$. For $q \in Q$, the value $\theta(q) = \max \{ m \mid (q,m) \in \dom(\sigma) \}$ is called \emph{threshold for $q$}. A $\sigma$-move from $(q,n)$ is an \emph{acceptor move} if $n = \theta(q)$, and a \emph{base move} if it is not an acceptor move and $n = \ext_\A(q)$. The \emph{reduced-memory strategy $\underline\rho$ induced by $\sigma$} is obtained by removing any non-acceptor, non-base moves from $\sigma$ and taking the carry-over memory-full strategy induced by the resulting skeleton. 
\end{definition}
Thus, $\underline\rho$-moves aimed at reaching accepting states are as prescribed by $\overline \rho$, however, moves aimed at increasing the available resources always use the moves prescribed by $\overline{\rho}$ when the \emph{smallest} possible amount of resources (that is, $\ext_\A(q)$) is available in $q$. The definition of $\underline\rho$ is inspired by the strategies of \cite{Chatterjee12}, which in non-accepting states combine an \emph{attractor strategy} (here the acceptor moves) with a \emph{good-for-energy strategy} (here the base moves). Our next result states that, if $\overline \rho$ can be used to exhibit an optimal run in $\A$ from each $(q,\ext_\A(q))$, then so can $\underline \rho$. We use an example to illustrate the proof idea.

\begin{example}
\label{ex:reduced}
Example~\ref{ex:strat-full} describes a skeleton strategy $\sigma$ which contains non-base, non-acceptor moves from states $(y,2)$ and $(y,4)$. Removing such moves from $\sigma$ yields the following skeleton:
{\small
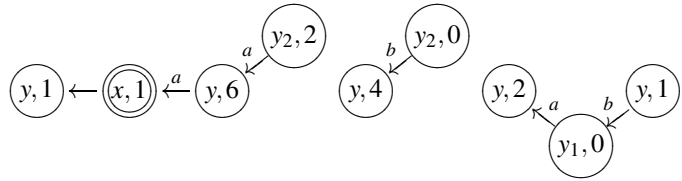
\begin{align*}
\entrymodifiers={++[o][F-]}\UseComputerModernTips\xymatrix@-2pc{
*{} &*{} &*{} &*{} &*{} & *{} & *{} & *{} & *{} & *{} & *{} & y_2,2 \ar[ddl]_-{a} & *{} & y_2,0 \ar[ddl]_-{b}  & *{} & *{} & *{} & *{} & *{} & *{} & *{} \\ 
*{}\\
y,1 & *{~~~~} & *{} & *{} & *++[o][F=]{x,1} \ar[llll] & *{~~~~} & *{} & *{} & *{} & *{} & y,6 \ar[llllll]_-{a} & *{}  & y,4 & *{} & y,2 & *{} & y,1 \ar[ddl]_-{b} \\
*{}\\
*{} &*{} &*{} &*{} & *{} & *{} & *{} & *{} & *{} & *{} & *{} & *{} & *{} & *{} & *{} & y_1, 0 \ar[uul]_-{a}
}
\end{align*}
}
\!\!with moves labelled so as to indicate their type (acceptor/base). In state $(y,1)$, a $\underline \rho$-conform play proceeds like a $\overline \rho$-conform play until $(y,2)$ is reached. Since $\ext_{\A}(y) < 2 < \theta(y)$, a $\underline \rho$-conform play from $(y,2)$ will repeat the base move for $y$ followed by the acceptor move for $y_1$, carrying over unused resources, until the acceptor move for $y$ is enabled and $(x,1)$ is reached:
{\small
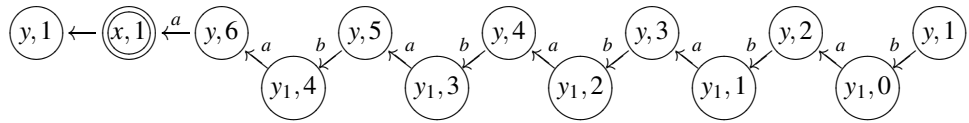
\begin{align*}
\entrymodifiers={++[o][F-]}\UseComputerModernTips\xymatrix@-2pc{
y,1 & *{~~~~} & *{} & *{} & *++[o][F=]{x,1} \ar[llll] & *{~~~~} & *{} & *{} & *{} & *{} & y,6 \ar[llllll]_-{a} & *{}  & y,5 \ar[ddl]_-{b} & *{} & y,4 \ar[ddl]_-{b} & *{} & y,3 \ar[ddl]_-{b} & *{}  & y,2 \ar[ddl]_-{b} & *{} & y,1 \ar[ddl]_-{b}\\
*{}\\
*{} & *{} & *{} & *{} & *{} & *{} & *{} & *{} *{} &*{} & *{} & *{} *{} & y_1, 4 \ar[uul]_-{a} & *{} & y_1, 3 \ar[uul]_-{a} & *{} & y_1, 2 \ar[uul]_-{a} & *{} & y_1, 1 \ar[uul]_-{a} & *{} & y_1, 0 \ar[uul]_-{a}
}
\end{align*}
}
\!\!To prove that $\underline \rho$ leads to an accepting state from any $(q,\ext_\A(q))$, we remove non-base, non-acceptor moves from $\sigma$ one by one in a particular order, and show that each removal maintains the winning property of the induced $\overline \rho$. Specifically, move $m$ from state $q$ is only removed if a $\overline \rho$-conform partial play starting in $(q,\ext_\A(q))$, which only visits non-accepting states, will not contain non-base, non-acceptor moves prior to $m$. 
For example, the move from $(y,4)$ is removed \emph{after} the move from $(y,2)$ is removed.
\end{example}

\begin{theorem}
\label{thm2}
Let $\sigma$ be a winning skeleton strategy for a WBAO $\A$, and let $\underline \rho$ be as in Definition~\ref{red-strat}. Then, any $\underline \rho$-conform run from $q$ with initial memory $\ext_\A(q)$ is accepting and has value $\ext_\A(q)$.
\end{theorem}
~\\[-40pt]
\begin{figure}[H]
\caption{Algorithm for computing attractor and base strategies for resource-aware WBAOs}
~\\[-5pt]
\small \noindent\fbox{
      \parbox{150mm}{
  {\bfseries Input:} resource-aware WBAO $ \A = (Q,\langle r,\gamma \rangle,\Omega)$
  
  {\bfseries Outputs:}  \,strategy $\sigma : Q_2 \to Q \cup \{\bot \}$; ~threshold $\theta : Q_1 \to \mathbb N \cup \{\bot\}$; 
  
  \qquad \qquad ~attractor strategy $\sigma_a : Q_1 \to Q \cup \{\bot\}$; ~good-for-energy strategy $\sigma_e : Q_1 \to Q \cup \{\bot\}$

~\\[-30pt]
  \qquad \qquad \quad 
  \begin{enumerate}
  \itemsep=-3pt
  \leftskip=1mm
  \item {\bfseries let} $e := Extent(2)$
  \item {\bfseries for} $q \in Q_2$ {\bfseries do}
  \item ~ ~ {\bfseries let} $\sigma(q) :=$ ~{\bfseries any} $(\lambda,q')$ {\bfseries \,s.t.} $(\gamma(q)(\iota_\lambda(q')) \ne \infty)  ~\wedge~ \left(e(q) = (\gamma(q)(\iota_\lambda(q')) + e(q')) \ominus r(q)\right)$
  \item {\bfseries endfor}
  \item {\bfseries for} $q \in Q_1$ {\bfseries do}
  \item ~ ~ {\bfseries let} $e(q) := \infty$; ~ 
  \item ~ ~ {\bfseries let} $\sigma_a(q) := \bot$
  \item {\bfseries endfor}
   \item {\bfseries repeat}
  \item ~ ~ {\bfseries let} $old := e$
  \item ~ ~ {\bfseries for} $q \in Q_1$ {\bfseries do}
  \item ~ ~ ~ ~ {\bfseries let} $e(q) := \min \{ (\gamma(q)(\iota_\lambda(q')) + old(q'))\varominus r(q) \mid  \gamma(q)(\iota_\lambda(q')) \ne \infty \}$
  \item ~ ~ ~ ~ {\bfseries let} $q_1 :=  $~{\bfseries any} $q'$ {\bfseries \,s.t.} $(\gamma(q)(\iota_\lambda(q')) \ne \infty  ~\wedge~ e(q) = \left(\gamma(q)(\iota_\lambda(q')) + old(q'))\varominus r(q)\right)$
  \item ~ ~ ~ ~ {\bfseries if} $e(q) \ne old(q)$ {\bfseries then}
  \item ~ ~ ~ ~ ~ ~ {\bfseries if} $\sigma_a(q) = \bot$ {\bfseries then}
  \item ~ ~ ~ ~ ~ ~ ~ ~ {\bfseries let} $\sigma_a(q) := q_1$ ~~~~~~~/* set attractor strategy */
  \item ~ ~ ~ ~ ~ ~ ~ ~ {\bfseries let} $\theta(q) := e(q)$
  \item ~ ~ ~ ~ ~ ~ {\bfseries else}
  \item ~ ~ ~ ~ ~ ~ ~ ~ {\bfseries let} $\sigma_e(q) := q_1$ ~~~~/* set good-for-energy strategy */
    \item ~ ~ ~ ~ ~ ~ {\bfseries endif}
  \item ~ ~ ~ ~ {\bfseries endif}
   \item ~ ~ {\bfseries endfor}
  \item {\bfseries until} $e = old$
  \end{enumerate}
  }  }
\label{fig:alg}
\end{figure}
Theorem~\ref{thm2} applies to the strategy computed by the algorithm outlined in Section~\ref{sec:memory-full}. In fact, this algorithm can be adapted to \emph{only} compute the resulting attractor and good-for-energy strategies -- see Figure~\ref{fig:alg} (where for succinctness we further assume that $\arity(\lambda) = 1$ for $\lambda \in \Lambda$). 
The correctness of the algorithm follows from Theorems~\ref{thm} and \ref{thm2}. We now analyse the complexity of computing the attractor strategy $\sigma_a$ and the good-for-energy strategy $\sigma_e$, assuming extents are already computed. 
If $N_i = |Q_i|$ and $E_i$ is the number of edges from states in $Q_i$, for $i\in \{1,2\}$, then (i) steps 2--8 take time O$(N_2 + E_2 + N_1)$, while (ii) steps 10--22 take time O$(N_1 + E_1)$ and are repeated at most $N_1 \times B$ times (in the worst case, each iteration changes only one $e(q)$). This gives an upper bound of O$((N_1+ E_1) \times N_1 \times B + N_2 + E_2)$ for the overall time complexity. While this is quadratic in the size of the automaton, the constant $B$ should not be ignored.

For our results to extend to \emph{tree} automata, Definition~\ref{def:skeleton-induced} must redistribute unused resources following a transition $q \to (\lambda,q_1,\ldots,q_n)$, so that an amount $>0$ is assigned to each $q_i$. (Without this, removing non-base, non-acceptor moves from $\sigma$ may lead to some branches of the resulting run never reaching accept states.) This can be done by moving to the (bounded) tropical rationals semiring, but \emph{only} once $\sigma$ is computed. This allows even a resource of $1$ to be redistributed among \emph{several} (finitely many) successors, ensuring that when a sequence of base/acceptor moves is repeated with \emph{more} resources, and this leads (on a branch) to a state seen before where the base move applies again, the resources have strictly increased, and a repetition of the same moves will again increase resources by at least the same amount. Then, the proof of Theorem~\ref{thm2} generalises: the only difference is that instead of a single play from $(y,\ext_\A(y))$ one must consider all plays arising from $\forall$ choices. The complexity in the case of tree automata remains similar, provided that $E_1$ now counts an "edge" $q \to (\lambda,q_1,\ldots,q_n)$ a number of times equal to $\arity(\lambda)$.

An extension to \emph{parity} automata can be obtained by observing that optimal accepting runs of a special form always exist in $\A$. Such runs can be described using (i) a \emph{lasso}, that is, a partial run of $\R_\A$ with root $(q_0,\ext_\A(q_0))$ and leaves of the form $(q,\ext_\A(q))$, and (ii) a \emph{loop}, that is, a winning strategy $\sigma$ in a \emph{sub-model} of $\R_\A$ that contains all the leaves in the lasso, with the additional property that a $\sigma$-conform play will \emph{never} visit an odd-parity state without later visiting a higher, even parity state. The lasso leaves thus identify points in the run from which states with odd parity larger than the highest infinitely occurring even parity on \emph{any} given branch do not occur. Once extents are known, a \emph{maximal} loop (i.e.~one which includes as many states of the form $(q,\ext_\A(q))$ as possible) can be computed as a greatest fixpoint, whereas a lasso from each of the remaining states of the form $(q,\ext_\A(q))$, with leaves belonging to the loop, can be computed as a least fixpoint. For the maximal loop, the associated strategy is computed similarly to the B\"uchi case, additionally recording a lower bound on the even parities guaranteed to be seen in the future on \emph{any} branch of the computation. However, this time several iterations are needed to successively discard states of the form $(q,\ext_\A(q))$ that cannot be part of a loop. For lassos, a backwards computation suffices. As in the B\"uchi case, the resulting strategies (either lasso or loop for any given state) can be simplified to consist of two moves only. The details of this extension, including a complexity analysis, are left for future work.

\section{Weighted Parity Games with Offsetting}
\label{sec:games}

We now define a game version of WPAOs, along with a suitable notion of extent. The definition below is parameterised by a semiring $(S,+,0,\bullet,1)$ subject to our earlier assumptions.
\begin{definition}
A \emph{weighted parity game with offsetting (WPGO)} is given by a $S \times (\FPow \circ \T_S \circ F)$-coalgebra $(Q,\langle r,\gamma \rangle)$ together with a \emph{parity map} $\Omega : Q \to \mathbb N$ with finite range. If $\ran(\Omega) = \{1,2\}$, we call $\G = (Q,\langle r,\gamma \rangle,\Omega)$ a \emph{weighted B\"uchi game with offsetting (WBGO)}.
\end{definition}
Thus, WPAOs are special cases of WPGOs, with each $\gamma(q)$ a singleton. 
The presence of the finite powerset functor $\FPow : \Set \to \Set$ gives $\forall$ additional power to increase the amount of initial resources required, by making a choice (in $\gamma(q)$) at each step of the computation.
\begin{remark}
While standard weighted parity games as used e.g.~in \cite{Chatterjee12, FZ14d,SWZ18} make no restrictions on the alternation of $\exists$ and $\forall$ moves, nor on the weights on $\forall$ moves or the parities of $\forall$ states, any weighted parity game which does not satisfy our constraints can be transformed into one that does, by adding intermediary states/transitions that offer no real choice\,/\,carry trivial weights.
\end{remark}
\begin{example}
When $S = (\{0,1\},\vee,0,\wedge,1)$ (so $\T_S = \FPow$) and $F = \Id$, ignoring offsets yields a (special type of) standard parity game, with $\forall$ states given by elements of $Q$ and inheriting parities from $Q$, $\exists$ states given by elements of $\gamma(q)$ ($q \in Q$) and having minimum parity, and $\forall$ and $\exists$ moves alternating.
\end{example}
\begin{definition}[Extent of WPGO]
\label{extent-game-def}
Let $\G = (Q,\langle r,\gamma \rangle,\Omega)$ be a WPGO with $\ran(\Omega) = \{1,\ldots,n\}$. For $k \in \ran(\Omega)$, let $\gamma_k : Q_k \to \FPow \T_S F Q$ and $r_k : Q_k \to S$ denote the restrictions of $\gamma$ and respectively $r$  to $Q_k := \{q \in Q \mid \Omega(q) = k\}$. The \emph{extent} $\ext_\G = [e_1,\ldots,e_n] : Q \to S$ of $\G$ is the solution of the following nested equational system, with the most significant equation being the last one:
\begin{eqnarray}
\label{eqn-extent-game}\begin{bmatrix}
u_1 & =_\mu  & ((\mu_1 \circ T_S (\bullet_F) \circ T_S F [u_1, \ldots, u_n])^\sharp \circ \gamma_1) \varoslash r_1\\
u_2 & =_\nu  & ((\mu_1 \circ T_S (\bullet_F) \circ T_S F [u_1, \ldots, u_n])^\sharp \circ \gamma_2) \varoslash r_2\\
& \vdots \\
u_n & =_\eta & ((\mu_1 \circ T_S (\bullet_F) \circ T_S F [u_1, \ldots, u_n])^\sharp \circ \gamma_n) \varoslash r_n
\end{bmatrix}
\end{eqnarray}
with $\eta = \mu$ ($\eta = \nu$) if $n$ is odd (resp.~even), with variables $u_k$ ranging over the poset $(S^{Q_k},\sqsubseteq)$ (and therefore $[u_1,\ldots,u_n] : Q \to S$), and with the first operands in the rhs$s$ pictured below:
\begin{align*}
\UseComputerModernTips\xymatrix@-0.75pc{Q_k \ar[r]^-{\gamma_k} & \FPow \T_S F Q \ar[rrrrrr]^-{(\T_SF[u_1,\ldots,u_n] ; \T_S (\bullet_F) ; \mu_1)^\sharp} & & & & & & S\\
& \T_S F Q \ar[rrr]^-{\T_S F [u_1, \ldots,u_n]} & & & \T_S F S \ar[rr]^-{\T_S (\bullet_F)} & & \T_S S = \T_S^2 1 \ar[r]^-{\mu_1} & \T_S 1 = S}
\end{align*}
Here, for $f : X \to S$, $f^\sharp : \FPow X \to S$ takes $Y \subseteq X$ to $\inf_{y \in Y}f(y)$.
\end{definition}
Thus, the only difference w.r.t.~WPAOs is the use of infima w.r.t.~$\sqsubseteq$ to account for the worst $\forall$ move. 
\begin{example}
For $S =  (\{0,1\},\vee,0,\wedge,1)$, $F = \Id$ and no offsetting, a state has extent $1$ iff $\exists$ has a winning strategy from it in the associated (standard) parity game. For $S = (\mathbb N_B^\infty,\min,\infty,+,0)$ and $F = \Id$, the extent gives the minimum resources required to exhibit an accepting run, irrespective of how $\forall$ moves.
\end{example}
\begin{remark}
\label{rem:chat}
The energy parity games of \cite{Chatterjee12}, which use integer weights and no offsetting, can be modelled as WPGOs over $(\mathbb N^\infty,\min,\infty,+,0)$ with $F = \Id$: transitions with negative weights become positively weighted, while positive weights are turned into offsets by suitably duplicating their target states along with any outgoing transitions from such states. Then, the \emph{initial credit problem} for energy parity games \cite{Chatterjee12} (of computing the minimum initial credit required for an accepting run whose energy level never drops below $0$) is equivalent to the problem of computing the extents of states in the associated WPGO.
\end{remark}
\begin{example}
\label{ex:simple-game}
Consider the WBGO on the left ($F = \Id$), where unlabelled transitions correspond to $\forall$ moves leading to different elements of $\T_S Q$ (namely $f$ and $g$). The state space is $Q = \{x,y_1,y_2\}$ and the offset function maps $x$ to $0$, $y_1$ to $1$ and $y_2$ to $4$. In states $x$ and $y_1$, $\forall$ can move to either $f$ or $g$. In state $f$, $\exists$ can move to $y_1$ or $y_2$ (with the latter move increasing the available resources in $f$) or to the accepting state $x$. The associated equational system (given on the right) has solution $(2,1,0,0,0)$. Intuitively, this is because by starting in state $x$ with $2$ initial resources, in the worst case $\forall$ moves to state $f$, where the only $\exists$ move that increases available resources is to $y_2$. Thereon $\exists$ can control whether to further increase the available resources (by moving to $y_2$ again) or move to the accept state (when $6$ resources are available in $f$). On the other hand, $1$ initial resource is required in $y_1$: together with the offset of $1$, this allows a move to $y_2$ when $\forall$ moves to $f$, or a move to $x$ that leaves $2$ available resources, when $\forall$ moves to $g$. 

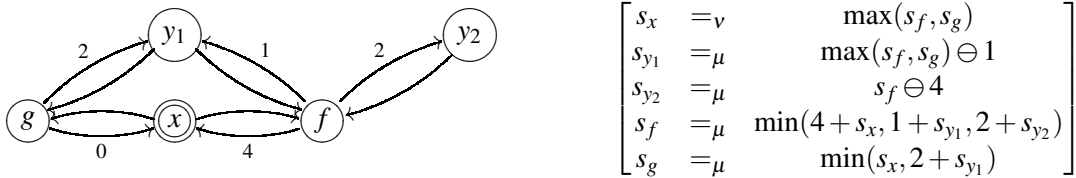
\begin{minipage}{0.5\textwidth}
$\entrymodifiers={++[o][F-]}\UseComputerModernTips\xymatrix@-1pc{
*{} & *{} & *{} & y_1 \ar@<-0.5ex>@/^0.5pc/[dlll] \ar@<0.5ex>@/_0.5pc/[drrr] & *{} & *{} & *{} & *{} & *{} & y_2 \ar@/^0.5pc/[dlll] \\
g  \ar@/^0.5pc/[urrr]^-{2} \ar@/_0.5pc/[rrr]_-{0} & *{} & *{} & *++[o][F=]{x} \ar@<0.5ex>@/_0.5pc/[lll] \ar@<-0.5ex>@/^0.5pc/[rrr] & *{} & *{} & f \ar@/^0.5pc/[urrr]^-{2}  \ar@/_0.5pc/[ulll]_-{1} \ar@/^0.5pc/[lll]^-{4}
}$
\end{minipage}
\begin{minipage}{0.5\textwidth}
$\begin{bmatrix}
s_x & =_\nu & \max(s_f,s_g) \\
s_{y_1} & =_\mu & \max(s_f,s_g) \varominus 1\\
s_{y_2} & =_\mu & s_f \varominus 4\\
s_f & =_\mu & \min(4+s_x,1+s_{y_1},2 + s_{y_2})\\
s_g & =_\mu & \min(s_x,2 + s_{y_1})
\end{bmatrix}$
\end{minipage}
\end{example}

\begin{remark}
\label{rem:run-game}
Notions of \emph{run} and \emph{value} of a run can be defined similarly to WPAOs: a run selects a \emph{single} $\forall$ move in states $q \in Q$, and explores \emph{all} $\forall$ moves in states $(\lambda,q_1,\ldots,q_n)$. However, Proposition~\ref{prop-ext} does not generalise, since $\forall$ moves in a \emph{specific} play may require fewer resources than specified by the extents.
\end{remark}

\section{Strategies in Resource-Aware Games}
\label{sec:strat-games}

We now fix $S = (\mathbb N^\infty,\min,\infty,+,0)$ and show how the results in Section~\ref{sec:strat-aut} extend to \emph{resource-aware WBGOs}: we define a resource game $\R_\G$ akin to the resource game $\R_\A$ and show that memoryless winning strategies for $\exists$ in $\R_\G$ yield strategies in $\G$ for  exhibiting resource-aware accepting runs with minimum initial resources (where this time, $\forall$ controls some of the choices made in defining the run). 
\begin{definition}[Resource Game]
\label{def:res-game}
The \emph{resource game $\R_\G$ associated to a WPGO $\G = (Q,\langle r,\gamma \rangle, \Omega)$} is a standard parity game with states given by $Q_\exists = \{(Y,s + r(q)) \in \T_S F Q \times \mathbb N_B \mid Y \in \gamma(q) \text{ for some } q \in Q \text{ with }s \ge \ext_\G(q)\}$ and $Q_\forall = \{(q,s) \in Q \times \mathbb N_B \mid s \ge \ext_{\G}(q)\} \cup F \{(q,s) \in Q \times \mathbb N_B \mid s \ge \ext_{\G}(q)\}$, parities inherited from $\G$ on $\forall$ states of the form $(q,s)$ with $s \ge \ext_{\G}(q)$, and equal to $\min \{\Omega(q) \mid q \in Q\}$ on all other states, and moves given by:
\begin{eqnarray*} 
(q,s)\UseComputerModernTips\xymatrix@-0.5pc{\ar[r]^-{\forall} &} (Y,s+r(q)) & \text{whenever}~ Y \in \gamma(q) \qquad \qquad \qquad \qquad \quad \qquad\\[-4pt]
(Y,s) \UseComputerModernTips\xymatrix@-0.5pc{\ar[r]^-{\exists} &} (\lambda,(q_1,s_1),\ldots,(q_n,s_n)) & \text{whenever}~ 
s \ge Y(\iota_\lambda(q_1,\ldots,q_n)) + s_1 + \ldots + s_n\,,\\[-4pt]
(\lambda,(q_1,s_1),\ldots,(q_n,s_n)) \UseComputerModernTips\xymatrix@-0.5pc{\ar[r]^-{\forall} &} (q_i,s_i) & \text{for all}~ i \in \{1,\ldots,n\}\,. \qquad \qquad \qquad \qquad \qquad 
\end{eqnarray*}
\end{definition}
The difference w.r.t.~Definition~\ref{res-model-def} is the addition of a new type of $\forall$ move, which is also where the offset values are added to the available resources.

\begin{definition}
A \emph{resource-aware play of a WPGO $\G$} is a play of $\R_\G$, viewed as a standard parity game.
\end{definition}
Such a play cycles through: $\forall$ making a choice from those specified by $\gamma$; $\exists$ choosing a specific transition out of the available ones; and in case of tree-like structure, $\forall$ choosing one particular branch resulting from the $\exists$ move. Then any memoryless strategy $\sigma$ for $\exists$ in $\R_\G$ yields a memory-full strategy $\rho$ for $\exists$ in $\G$, defined as in Definition~\ref{rho}. Moreover, Theorem~\ref{thm1} generalises to WBGOs.

\begin{theorem}
\label{thm3}
Let $\G = (Q,\langle r,\gamma \rangle, \Omega)$ be a WPGO with associated resource game $\R_\G$. 
If $\sigma : Q_\exists \to Q_\forall$ is winning for $\exists$ in $\R_\G$, the induced strategy $\rho$ can be used to exhibit a resource-aware accepting run of $\G$ from $q \in Q$ with initial memory $\ext_\G(q)$. (Recall from Remark~\ref{rem:run-game} that such a run requires input from $\forall$.)
\end{theorem}
\begin{example}
\label{ex:sigma-game}
A strategy $\sigma$ for the game in Example~\ref{ex:simple-game} is shown below. Note that this only prescribes 

\begin{minipage}{0.475\textwidth}
{\small 
$\entrymodifiers={++[o][F-]}\UseComputerModernTips\xymatrix@-2pc{
*{} & *{} & *{} & *{} & *{} & g,1 \ar[ddlllll]\\
*{} & *{} & *{} & *{} & *{} & *{} & *{} & *{} & *{} & *{~} & y_1,5 \ar@{-->}[dlllll] \ar@{-->}[ulllll] & *{} & *{} & *{} & *{} & *{} & *{} & *{} & *{~} & y_1,3 \ar@{-->}[dlllll] \ar@{-->}[ullllllllllllll] & *{} & *{} & *{} & *{} & *{} & *{} & *{} & *{~} & y_1,1 \ar@{-->}[dlllll] \ar@{-->}@/_1pc/[ulllllllllllllllllllllll]  \\ 
*++[o][F=]{x,2} & *{~~~} & *{} & *{} & *{} & f,6 \ar[lllll] & *{} & *{} & *{} & *{} & *{} & *{} & *{} & *{} & f,4 \ar[dllll] & *{} & *{} & *{} & *{} & *{} & *{} & *{} & *{} & f,2 \ar[dllll] \\
*{} & *{} & *{} & *{} & *{} & *{} & *{} & *{} & *{} & *{} & y_2,2 \ar@{-->}[ulllll] & *{} & *{} & *{} & *{} & *{} & *{} & *{} & *{} & y_2,0 \ar@{-->}[ulllll] 
}$
}\\[1pt]
\end{minipage} 
\begin{minipage}{0.48\textwidth}
~\\[-3.5pt]
\noindent moves (solid arrows) in states $f$ and $g$. Dashed arrows correspond to $\forall$ moves, and the values shown for $y_1$ and $y_2$ correspond to the \emph{worst} $\forall$ move from each state; a different $\forall$ move may result in more resources available in $f$/$g$ than shown in the target states (as is e.g.~the case for $\UseComputerModernTips\xymatrix@-0.75pc{(y_1,1) \ar@{-->}[r] & (g,1)}$).
\end{minipage}
\end{example}
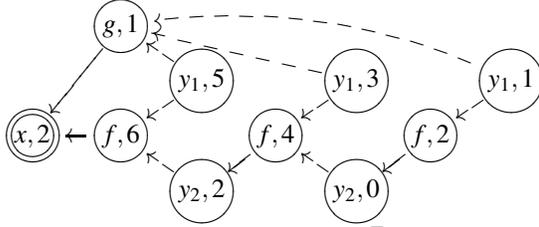
\emph{Reduced-memory strategies} for $\exists$ can be defined for WPGOs just as in the automata case; they arise from \emph{skeleton strategies}, also defined as before. For example, the reduced memory strategy induced by $\sigma$ of Example~\ref{ex:sigma-game} involves discarding the move from $(f,4)$. Moreover, Theorem~\ref{thm2} generalises to games.
\begin{theorem}
\label{thm4}
Let $\sigma$ be a winning skeleton strategy for a WBGO $\G$, inducing a reduced-memory strategy $\underline \rho$. Then $\underline \rho$ is winning in $\G$ (in the sense of Theorem~\ref{thm3}) from any $q \in Q$ with initial memory $\ext_\A(q)$. 
\end{theorem}
The algorithm in Figure~\ref{fig1} generalises straightforwardly to WPGOs: the computation of the extent values $e(q)$ with $q \in Q$ (line 9) additionally quantifies over \emph{all} $\forall$ moves in $q$, taking the supremum (w.r.t.~$\le$) of the resulting values. With this change, the algorithm in Figure~\ref{fig:alg} generalises to WBGOs, updating either $\sigma_a(q)$ or $\sigma_e(q)$ whenever $e(q)$ is updated. The complexity of computing both extents and reduced-memory strategies is similar to WBAOs; the only difference is that $E_i$ now counts all ways to reach a state in $Q$ from a state in $Q_i$ in one step (one application of $\gamma$) via an intermediary $\forall$-move, for $i \in \{1,2\}$.

\section{Conclusions and Future Work}
\label{sec:concl}

We studied the optimal strategy synthesis problem for a notion of resource-aware parity automaton/game that models both resource gain and resource usage, and showed that memory-full strategies consisting of two moves only per state suffice. We focused on the B\"uchi case, leaving the extension to parity automata/games for future work. 
While our results are similar to existing ones (e.g.~\cite{Chatterjee12}), they are applicable to a wider class of qualitative behaviours which includes not only words but also trees. This is made possible by a coalgebraic treatment of quantitative automata/games, which also opens the road for further extensions of our approach. Future work will study more complex synthesis problems, applicable to systems with a variety of quantitative features (e.g.~the synthesis of a resource-aware component operating in a probabilistic environment, with guarantees on the likelihood with which available resources suffice).
\newpage
\nocite{*}
\bibliographystyle{eptcs}
\bibliography{gandalf}

\end{document}